\documentclass[11pt]{article}

\usepackage[T1]{fontenc}
\usepackage[utf8]{inputenc}
\usepackage{lmodern}
\usepackage{microtype}
\input{glyphtounicode}
\usepackage{mathtools}
\mathtoolsset{mathic=true}

\usepackage{amsmath}
\usepackage{amsfonts}
\usepackage{amssymb}
\usepackage{amsthm}

\usepackage{fullpage}
\usepackage{hyperref}
\usepackage{authblk}
\usepackage{graphicx}
\usepackage{caption}
\usepackage{subcaption}

\usepackage[ruled,vlined,linesnumbered]{algorithm2e}

\newtheorem{theorem}{Theorem}
\newtheorem{lemma}[theorem]{Lemma}

\newcommand{\calO}{\mathcal{O}}
\newcommand{\ie}{i.e.}

\DeclareMathOperator{\pb}{pb}
\DeclareMathOperator{\spb}{spb}
\newcommand{\csMod}{\ \mathrm{mod} \ }

\begin{document}

\title{Equivalence between pathbreadth and strong pathbreadth}

\author[1,2]{Guillaume Ducoffe}
\affil[1]{\small National Institute for Research and Development in Informatics, Romania}
\affil[2]{\small The Research Institute of the University of Bucharest ICUB, Romania}

\author[3]{Arne Leitert}
\affil[3]{\small Department of Computer Science, Central Washington University}

\date{}

\maketitle

\begin{abstract}
We say that a given graph $G = (V, E)$ has \emph{pathbreadth} at most $\rho$, denoted $\pb(G) \leq \rho$, if there exists a Roberston and Seymour's path decomposition where every bag is contained in the $\rho$-neighbourhood of some vertex.
Similarly, we say that $G$ has \emph{strong pathbreadth} at most $\rho$, denoted $\spb(G) \leq \rho$, if there exists a Roberston and Seymour's path decomposition where every bag is the complete $\rho$-neighbourhood of some vertex.
It is straightforward that $\pb(G) \leq \spb(G)$ for any graph $G$.
Inspired from a close conjecture in [Leitert and Dragan, COCOA'16], we prove in this note that $\spb(G) \leq 4 \cdot \pb(G)$.
\end{abstract}

We refer to~\cite{BoM08} for any undefined graph terminology.
Graphs in this study will be finite, simple, \emph{connected} and unweighted.
Our purpose in this note is to relate two pathlikeness invariants, first introduced in~\cite{DKL17,LeD16}.
Specifically, a (Robertson and Seymour's) \emph{path decomposition} of a given graph $G = (V, E)$ is any sequence $\big( X_1, X_2, \ldots, X_p \big)$ of subsets of $V$, called \emph{bags}, that satisfies the following three properties:
\begin{enumerate}
\item Every vertex $x \in V$ is contained in at least one bag;
\item Every edge $xy \in E$ has its two ends contained in at least one common bag;
\item For every $x \in V$, the bags that contain $x$ induce a consecutive subsequence.
\end{enumerate}
The width of a path decomposition is equal to the largest size of its bags minus one.
The pathwidth of a graph $G$ is the minimum possible width over its path decompositions.
Pathwidth is often used in parameterized complexity as it has many algorithmic applications.
Motivated by the efficient resolution of routing and distance-related problems on graphs~\cite{DrL04}, we rather focus in this note on the metric properties of the bags instead of their size.

The \emph{breadth} of a path decomposition is equal to the smallest integer $\rho$ such that every bag is contained in the $\rho$-neighbourhood of some vertex (this vertex may not be in the bag).
The \emph{pathbreadth} of a graph $G$, denoted $\pb(G)$, is the minimum possible breadth over all its path decompositions.
We stress that bounded-pathbreadth graphs comprise the interval graphs and the convex bipartite graphs, that are two important graph classes with unbounded pathwidth.
However, computing the pathbreadth of a given graph is an NP-hard problem~\cite{DLN16}.

A slightly more amenable parameter than pathbreadth -- unfortunately still NP-hard to compute~\cite{Duc18} -- is \emph{strong pathbreadth}, defined as follows.
The strong pathbreadth of a graph $G$, denoted $\spb(G)$, is the minimum integer $\rho$ such that there exists a path decomposition of $G$ where all bags are the \emph{complete} $\rho$-neighbourhood of some vertex.
Note that we clearly have $\pb(G) \leq \spb(G)$.
It is natural to ask whether, conversely, there exists a universal constant $c$ such that $\spb(G) \leq c \cdot \pb(G)$.
In fact, a similar question was asked in~\cite{LeD16} for the related parameters treebreadth and strong treebreadth (defined using the more general object of tree decompositions).
In this note, we answer positively to this question for pathbreadth and strong pathbreadth.
Namely, we prove the following result:

\begin{theorem}\label{thm:main}
For every graph \( G \), we have \( \pb(G) \leq \spb(G) \leq 4 \cdot \pb(G) \).
\end{theorem}

To prove Theorem~\ref{thm:main}, we describe in Algorithm~\ref{algo:constructPhi} below how to to construct a path decomposition with strong breadth at most~$4 \cdot \pb(G)$ for a given graph~$G$.
The \emph{eccentricity} of a shortest path~$P$ in~$G$ is defined in what follows as the maximum distance between any vertex in~$V$ and a closest vertex in~$V(P)$.

\begin{algorithm}[!ht]
\caption
{%
    Computes a path decomposition with strong breadth~$2 \lambda$ for a given graph and a given shortest path with eccentricity~$\lambda$.
}
\label{algo:constructPhi}

\KwIn
{%
    A graph~$G = (V, E)$ and a shortest path~$P = (v_0, v_1, \ldots, v_\ell)$ with eccentricity~$\lambda$.
}

\KwOut
{%
    A path decomposition~$\Phi$ for~$G$ with strong breadth~$2 \lambda$ and the centers~$Q$.
}

Let $P = (v_0, v_1, \ldots, v_\ell)$, set $L := \lfloor \ell / 2 \lambda \rfloor$, and set $\delta := \big \lfloor (\ell \csMod 2 \lambda) / 2 \big \rfloor$.
\label{line:compLdelta}

\For
{%
    \( i := 0 \) \KwTo \( L \)
}
{%
    Let $j = 2 \lambda \cdot i + \delta$ and set $q_i := v_j$.
    \label{line:compQi}

    Compute the bag~$B_i := N^{2 \lambda}[q_i]$ by performing a BFS which starts at~$q_i$ and is limited to distance~$2 \lambda$.
    \label{line:compBi}
}

Output $Q = \{ q_1, q_2, \ldots, q_L \}$ and $\Phi = (B_1, B_2, \ldots, B_L)$.
\end{algorithm}

\begin{lemma}
Algorithm~\ref{algo:constructPhi} constructs a path decomposition~\( \Phi \) for~\( G \) with strong breadth~\( 2 \lambda \) in linear time.
\end{lemma}

\begin{proof}
For the first part of the proof, we show that the sequence
\[
    \Phi = \Big( N^{2 \lambda}[q_0], N^{2 \lambda}[q_1], \ldots, N^{2 \lambda}[q_L] \Big)
\]
constructed by the algorithm is a path decomposition for~$G$.
In order to prove this claim, it suffices to prove that $\Phi$ satisfies all the properties of a path decomposition.
Clearly, in that case, $\Phi$ has strong breadth~$2 \lambda$.

\begin{itemize}
\item
We first show that each vertex is contained in a bag.
Observe that, by construction of~$Q$, $d(q_i, q_{i + 1}) = 2 \lambda$ for all $i < L$ and $\min \big \{ d(v_0, q_0), d(q_L, v_\ell) \big \} \leq \lambda$.
It follows that $Q$ is a $\lambda$-dominating set of~$P$.
Since $P$ is a $\lambda$-dominating path for~$G$, we obtain that
$\bigcup_{i = 0}^L N^{2 \lambda}[q_i] \supseteq \bigcup_{i = 0}^\ell N^\lambda[v_i] = V$.
Hence, every vertex is contained in a bag.
\item
Next, we show that each edge is contained in a bag of~$\Phi$.
Let $xy$ be an arbitrary edge.
Suppose for the sake of contradiction that, for every $q_i \in Q$, we have $\{ x, y \} \nsubseteq N^{2 \lambda}[q_i]$.
Since $P$ $\lambda$-dominates $G$, $P$ contains two vertices $x'$ and~$y'$ with $d(x, x') \leq \lambda$ and $d(y, y') \leq \lambda$.
Note that $x' \neq y'$.
Assume that $d(x, x') < \lambda$ or that there is a vertex $q_i \in Q$ with $d(q_i, x') < \lambda$.
Then, $d(q_i, y) \leq d(q_i, x') + d(x, x') + 1 < 2 \lambda + 1$ and, hence, $\{ x, y \} \subseteq N^{2 \lambda}[q_i]$.
It follows that $d(x, x') = d(y, y') = \lambda$, and that $d(q_i, x')$ and $d(q_i, y')$ are at least~$\lambda$ for each~$q_i \in Q$.
Recall that the distance between two consecutive vertices in~$Q$ is exactly~$2 \lambda$.
Hence, $x'$ and~$y'$ are respectively in the middle of two consecutive vertices in~$Q$ with equal distance~$\lambda$ to them.
Note that, since $P$ is a shortest path, $d(x', y') \leq d(x, x') + 1 + d(y, y') \leq 2 \lambda + 1$.
Therefore, there is a vertex~$q_i \in Q$ such that $d(x', q_i) = d(q_i, y') = \lambda$, and $d(x', y') = 2 \lambda$ (otherwise, $d(x', y')$ would be larger than~$2 \lambda + 1$ or $x'$ and~$y'$ would be equal).
But then, $d(q_i, x) \leq d(q_i, x') + d(x, x') \leq 2 \lambda$, $d(q_i, y) \leq d(q_i, y') + d(y, y') \leq 2 \lambda$, and, therefore, $\{ x, y \} \subseteq N^{2 \lambda}[q_i]$.
This contradicts with our assumption that no such~$q_i$ exists.
Altogether combined, it follows that every edge is contained in a bag.
\item
It remains to show that, for each vertex, the bags containing it are consecutive.
Let $x$ be an arbitrary vertex of~$G$.
Suppose for the sake of contradiction that there exist two vertices $q_j, q_k \in Q$ with $j < k - 1$ such that $x \in N^{2 \lambda}[q_j] \cap N^{2 \lambda}[q_k]$ and $x \notin N^{2 \lambda}[q_i]$ for every $i \in \{ j + 1, \ldots, k - 1 \}$.
Observe that $d(q_j, q_k) \leq d(q_j, x) + d(x, q_k) \leq 4 \lambda$.
Since $P$ is a shortest path and $j < k - 1$, we deduce that $d(q_j, x) = d(q_k, x) = 2 \lambda$ and, hence, that there is a vertex~$q_i \in Q$ with $d(q_j, q_i) = d(q_k, q_i) = 2 \lambda$, \ie, $q_i$ is between $q_j$ and~$q_k$.
However, since $P$ is a $\lambda$-dominating path for~$G$, $P$ contains a vertex~$x'$ with $d(x, x') \leq \lambda$.
By the triangle inequality, we have $d(q_j, x') \leq d(q_j, x) + d(x, x') \leq 3 \lambda$, and in the same way $d(x', q_k) \leq 3 \lambda$.
Since $d(q_j, q_k) = 4 \lambda$, this implies that $x'$ is between $q_j$ and~$q_k$ in~$P$ and that $d(x', q_i) \leq \lambda$.
Thus, $x \in N^{2 \lambda}[q_i]$ which contradicts with our original assumption.
Therefore, all the bags that contain~$x$ induce a consecutive subsequence.
\end{itemize}
Overall, $\Phi$ satisfies all the properties of a path decomposition, thereby proving the claim.

We now show that $\Phi$ can be constructed in linear time.
Calculating $L$ and~$\delta$ (line~\ref{line:compLdelta}) as well as determining all vertices~$q_i$ (line~\ref{line:compQi}) can easily be done in linear time.
To show that constructing all bags (line~\ref{line:compBi}) requires linear time in total, we recall that the distance between two consecutive vertices in~$Q$ is exactly~$2 \lambda$.
Thus, if a vertex~$v$ is contained in the bags $B_i = N^{2 \lambda}[q_i]$ and~$B_j = N^{2 \lambda}[q_{i + 2}]$ for some~$i$, then $d(q_i, v) = d(q_{i + 2}, v) = 2 \lambda$.
That is, $v$ is on the boundary of the bags $B_i$ and~$B_j$.
As a result, each vertex of~$G$ can be in at most three bags and each edge of~$G$ is in at most two bags.
Therefore, performing a BFS which is limited to distance~$2 \lambda$ on each vertex~$q_i$ requires at most $\calO(3n + 2m)$ time, \ie, line~\ref{line:compBi} runs in total linear time.
\end{proof}

Based on Algorithm~\ref{algo:constructPhi}, we can now prove Theorem~\ref{thm:main}.

\begin{proof}[Proof of Theorem~\ref{thm:main}]
Note that each graph~$G$ contains a shortest path~$P$ with eccentricity~$\lambda \leq 2 \pb(G)$~\cite{DKL17}.
Performing Algorithm~\ref{algo:constructPhi} on~$P$ then creates a path decomposition for~$G$ with strong breadth~$2 \lambda \leq 4 \pb(G)$.
It follows that $\spb(G) \leq 4 \pb(G)$ for any graph~$G$, thereby proving Theorem~\ref{thm:main}.
\end{proof}

\paragraph{Algorithmic applications.}
An \emph{asteroidal triple} in a given graph~$G$ is an independent set of size three in~$G$ such that each pair of two vertices in the triple is joined by a path that avoids the closed neighbourhood of the third one.
A graph is called \emph{AT-free} if it does not have any asteroidal triple.
It is known that each AT-free graph~$G$ has a vertex pair~$x, y$ such that each path from $x$ to~$y$ has eccentricity~$1$; such a pair can be found in linear time~\cite{COS97}.
We can now compute a shortest path~$P$ from $x$ to~$y$ and perform Algorithm~\ref{algo:constructPhi} on~$P$.
The output is a path decomposition with strong breath~$2$ for~$G$.
Therefore, we get the following (improving a result from~\cite{DKL17}):

\begin{theorem}
    \label{theo:ATfree}
If a graph~\( G \) is AT-free, a path decomposition for~\( G \) with strong breath~\( 2 \) can be computed in linear time.
\end{theorem}

Note that a decomposition as constructed by Algorithm~\ref{algo:constructPhi} is not necessarily optimal for all AT-free graphs.
See Figure~\ref{fig:ATfree} below for an example.

\begin{figure}
    [!htb]
    \centering
    \includegraphics{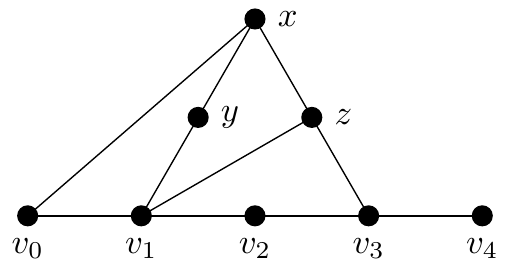}
    \caption
    {%
        An AT-free graph~$G$.
        The decomposition $\big( N^2[v_0], N^2[v_2], N^2[v_4] \big)$ (as constructed by Algorithm~\ref{algo:constructPhi}) has strong breadth~$2$.
        The decomposition $\big( N[v_0], N[y], N[z], N[v_2], N[v_4] \big)$, however, has strong breadth~$1$.
        \label{fig:ATfree}
    }
\end{figure}

Algorithm~\ref{algo:constructPhi} also allows to approximate the strong pathbreadth of a given graph with a constant approximation factor.

\begin{lemma}
Let~\( k \) be the minimum eccentricity of a shortest path in a graph \( G \).
If there is an algorithm that finds a shortest path in~\( G \) with eccentricity~\( \phi k + \psi \) in \( \calO \big( T(G) \big) \) time, then there is an algorithm to construct a path decomposition for~\( G \) with strong breadth at most~\( 4 \phi \spb(G) + 2 \psi \) in \( \calO \big( T(G) + n + m \big) \) time.
\end{lemma}

\begin{proof}
Let $P$ be a shortest path in~$G$ such that $P$ has eccentricity~\( \phi k + \psi \) and let $\Phi$ be a path decomposition constructed by performing Algorithm~\ref{algo:constructPhi} on~$P$.
By construction, $\Phi$ has strong breadth~$2(\phi k + \psi)$.
Recall that each graph~$G$ contains a shortest path with eccentricity~$k \leq 2 \pb(G) \leq 2 \spb(G)$.
Therefore, $\Phi$ has strong breadth~$2(\phi k + \psi) \leq 4\phi \spb(G) + 2 \psi$.
Since Algorithm~\ref{algo:constructPhi} runs in linear time, it takes in total $\calO \big( T(G) + n + m \big)$ time to construct~$\Phi$.
\end{proof}

Note that there is an $\calO \big( n^3 \big)$-time algorithm which finds a $2$-approximation for the Minimum Eccentricity Shortest Path problem~\cite{DraganLeiter2017}, and there is a linear-time algorithm which finds a $3$-approximation~\cite{BDP16}.
Therefore, we can conclude as follows:

\begin{theorem}
The strong pathbreadth of a given graph can be approximated by a factor~\( 8 \) in \( \calO \bigl( n^3 \bigr) \) time and by a factor~\( 12 \) in linear time.
\end{theorem}

\section*{Acknowledgements}

This work was supported by the Institutional research programme PN 1819 ``Advanced IT resources to support digital transformation processes in the economy and society - RESINFO-TD'' (2018), project PN 1819-01-01 ``Modeling, simulation, optimization of complex systems and decision support in new areas of IT\&C research'', funded by the Ministry of Research and Innovation, Romania.

\bibliographystyle{abbrv}
\bibliography{pb-spb-rel}

\end{document}